\documentclass[conference]{IEEEtran}
\IEEEoverridecommandlockouts
\usepackage{cite}
\usepackage{amsmath,amssymb,amsfonts}
\usepackage{algorithm,algorithmic}
\usepackage{graphicx}
\usepackage{mathtools}
\DeclarePairedDelimiter\ceil{\lceil}{\rceil}
\DeclarePairedDelimiter\floor{\lfloor}{\rfloor}
\usepackage{graphicx}
\usepackage{amsthm}
\newtheorem{theorem}{Theorem}
\newtheorem{lemma}[theorem]{Lemma}

\usepackage{textcomp}
\usepackage{xcolor}
\def\BibTeX{{\rm B\kern-.05em{\sc i\kern-.025em b}\kern-.08em
    T\kern-.1667em\lower.7ex\hbox{E}\kern-.125emX}}
\begin{document}

\title{An Upper Bound for Sorting $R_n$ with LRE\\}

\author{\IEEEauthorblockN{Sai Satwik Kuppili\IEEEauthorrefmark{1}
and Bhadrachalam Chitturi\IEEEauthorrefmark{2}}

\IEEEauthorblockA{\IEEEauthorrefmark{1} Dept. of Computer Science and Engineering,
Amrita University\\
Amritapuri, India\\
\IEEEauthorblockA{\IEEEauthorrefmark{2}Dept. of Computer Science,
University of Texas at Dallas\\
Texas, USA\\}
\IEEEauthorblockA{\IEEEauthorrefmark{1}VMware Inc.
Pune, India\\}
Email: \IEEEauthorrefmark{1}satwik.kuppili@gmail.com,
\IEEEauthorrefmark{2}chalam@utdallas.edu}}
\maketitle
\maketitle
\footnote{This article is submitted to 10th International Advanced Computing Conference, IACC-2020.}
\begin{abstract}
A permutation $\pi$ over alphabet $\Sigma = {1,2,3,\ldots,n}$, is a sequence  where every element $x$ in $\Sigma$ occurs exactly once. $S_n$ is the symmetric group consisting of all permutations of length $n$ defined over $\Sigma$. 
 $I_n$ = $(1, 2, 3,\ldots, n)$ and  $R_n =(n, n-1, n-2,\ldots, 2, 1)$ are identity (i.e. sorted) and reverse permutations respectively. 
 An operation, that we call as an $LRE$ operation, has been defined in OEIS with identity A186752.
 This operation is constituted by three generators: left-rotation, right-rotation and transposition(1,2). We call transposition(1,2) that swaps the two leftmost elements as $Exchange$.
 The minimum number of moves required to transform $R_n$ into $I_n$ with $LRE$ operation are known for $n \leq 11$ as listed in OEIS  with sequence number A186752. 
 For this problem no upper bound is known. OEIS sequence A186783 gives the conjectured diameter of the symmetric group $S_n$ when generated by $LRE$ operations \cite{oeis}.
 The contributions of this article are: (a) The first non-trivial upper bound for the number of moves required to sort $R_n$ with $LRE$; (b) a tighter upper bound for the number of moves required to sort $R_n$ with $LRE$; and (c) the minimum number of moves required to sort $R_{10}$ and $R_{11}$ have been computed. Here we are computing an upper bound of the diameter of Cayley graph generated by $LRE$ operation. Cayley graphs are employed in computer interconnection networks to model efficient parallel architectures. The diameter of the network corresponds to the maximum delay in the network.
\end{abstract}


\section{Introduction}
The following problem is from OEIS  with sequence number A186752: ``Length of minimum representation of the permutation $[n,n-1,\ldots,1]$ as the product of transpositions $(1,2)$ and left and right rotations $(1,2,\ldots,n)$. \cite{oeis}." We call this operation as $LRE$. 
$LRE$ operation consists of following three generators: (i) $Left Rotate$ that cyclically shifts all elements to left by one position, (ii) $Right Rotate$ that cyclically shifts all elements to right by one position and (iii) $Exchange$ that swaps the leftmost two elements of the permutation. The mentioned operations are abbreviated as $L$, $R$ and $E$ respectively.
$R_n$ denotes $(n,n-1,....,2,1)$ whereas $I_n$ denotes the sorted order or identity permutation: $(1,2,\ldots, n)$. Sorting a permutation $\pi$ in this article refers to transforming $\pi$ into $I_n$ with $LRE$ operation. The alphabet is $\Sigma= (1, 2, 3, \ldots, n)$. 
 \cite{KCN19, KC20} studied a more restricted version of this problem, i.e. $LE$ operation where the operation $R$ is disallowed and appears in OEIS with sequence number A048200 \cite{oeis}. We note that the results of  \cite{KCN19, KC20} are applicable to $RE$ operation (that has not been studied) due to symmetry.
 We seek to obtain an upper bound on the length of generator sequence that transforms $R_n$ with $LRE$ into $I_n$. \\
The optimum number of moves to sort $R_n$ with $LRE$ are known only for $n\leq 11$ ($n=10$ and $n=11$ are our contributions). We give the first non-trivial upper bound to sort $R_n$ with $LRE$.

Let $\pi[1 \cdots n]$ be the array containing the input permutation. The element at an index $i$ is denoted by $\pi[i]$. Initially for all $i$, $\pi[i]= R_n[i]$. 
We define a permutation $K_{r,n} \in S_n$ as follows. The elements $n-(r-1),n-(r-2),\ldots n$ are in sorted order i.e. the largest $r$ elements of $\Sigma$ are in sorted order. $K_{r,n}$ is obtained by concatenating sublists $(n-(r-1),n-(r-2),\ldots n)$ and $(n-r,n-(r+1), \ldots 3,2,1)$. Therefore a permutation $K_{r,n}$ can be denoted as follows $(n-(r-1), n-(r-2), \ldots n,n-r,n-(r+1), \ldots 3,2,1)$. \\
Therefore, $K_{1,n}$ is $(n, n-1, \ldots 3,2,1)$ which is $R_n$ and $K_{n,n}$ $(1,2,\ldots, n-1, n)$ which is $I_n$. 
Let $LE$ denote execution of Left-Rotate move followed by a Exchange move and $RE$ denote execution of Right-Rotate move followed by a Exchange move. Further, let $(LE)^p$ and $(RE)^p$ be $p$ consecutive executions of $RE$ and $RE$ respectively. Similarly, let $L^p$ and $R^p$ be $p$ consecutive executions of $L$ and $R$ respectively. 
\subsection{Background}
A Cayley graph $\Gamma$ defined on Symmetric group $S_n$, corresponding to an operation $\Psi$ with a generator set $G$ has $n!$ vertices each vertex corresponding to a unique permutation. 
An edge in $\Gamma$ from a vertex $u$ to another vertex $v$ indicates that there exits a generator $g \in G$ such that when $g$ is applied to $u$ one obtains $v$. Applying a generator is called as making a \emph{move}. 
An upper bound of $x$ moves to sort any permutation in $S_n$ indicates that the diameter of $\Gamma$ is at most $x$. An exact upper bound equals the diameter of $\Gamma$. Cayley graphs have many properties that render them apt for computer interconnection networks \cite{A:K:1989,LV:etal:1993}. Various operations to sort permutations have been posed that are of theoretical and practical interest \cite{LV:etal:1993}. 

Jerrum showed that when the number of generators is greater than one, the computation of minimum length of sequence of generators to sort a permutation is intractable \cite{Je85}. $LRE$ operation has three generators and the complexity of transforming one permutation in to another with $LRE$ unknown. Exchange move is a reversal of length two, in fact it is a prefix reversal of length two. 

For sorting permutations with (unrestricted) prefix reversals the operation that has $n-1$ generators, the best known upper bound is $18n/11+O(1)$ \cite{Ch09}. In $LRE$ operation, both left and right rotate cyclically shifts the entire permutation. In contrast, \cite{Fe10} an extended bubblesort is considered, where an additional swap is allowed between elements in positions 1 and $n$. 
We call an operation say $\Psi$ symmetric if for any generator of $\Psi$ its $inverse$ is also in $\Psi$. Exchange operation is inverse of itself whereas left and right rotate are inverses of one another, thus, $LRE$ is symmetric. Both $LE$ and $LRE$ are restrictive compared to the other operations that are studied in the context of genetics e.g. \cite{C11}. 
Research in the area of Cayley graphs has been active. Cayley graphs are studied pertaining to their efficacy in modelling a computer interconnection network, their properties in terms of diameter, presence of greedy cycles in them etc.  \cite{M17,E18,G18}. Efficient computation of all distances, some  theoretical properties of specific Cayley graphs, and efficient counting of groups of permutations in $S_n$ with related properties have been recently studied \cite{cdas2018,T18, bj2019,Ch20}.

\section{Algorithm LRE}
Algorithm $LRE$ sorts $R_n$ in stages. It first transforms $R_n$ which is identical to $K_{1,n}$ into $K_{2,n}$ by executing an $E$ move. Subsequently, $K_{i+1,n}$ is obtained from $K_{i,n}$ by executing the moves specified by Lemma \ref{l1}. Thus, eventually we obtain $K_{n,n}$ which is identical to $I_n$. Pseudo Code for the Algorithm $LRE$ is shown below.\\\\
\noindent\textbf{Algorithm LRE}\\
{Input: $R_n$. Output: $I_n$.\\
Initialization:$\forall{i}$ $\pi[i]=R_n[i]$.\\ All moves are executed on $\pi$.
\begin{algorithm}
\begin{algorithmic}[1]
\caption{Algorithm LRE}
\FOR{$r\in(1,\ldots,n-2)$}
\IF{$r=(n-2)$}
\STATE Execute $R^2$
\STATE Execute E move
\ELSE
\STATE Execute $(L)^{r-1}$
\STATE Execute E move
\STATE Execute $(RE)^{r-1}$
\ENDIF
\ENDFOR
\end{algorithmic}
\end{algorithm}{}

}
\subsection{Analysis}
\begin{lemma}
\label{l1}
The number of moves required to obtain $K_{r+1,n}$ from $K_{r,n} \forall r \in (1,\ldots, n-3)$ is $3r-2$.
\end{lemma}
\begin{proof} According to the definition, $K_{r,n}$ is \\
$(n-(r-1), n-(r-2), \ldots n-1, n, n-r, n-(r+1), n-(r+2), \ldots 3,2,1)$. \\
Executing $L^{r-1}$ on $K_{r,n}$ yields \\
$(n, n-r, n-(r+1), n-(r+2), \ldots 3, 2, 1, n-(r-1), n-(r-2), \ldots n-1)$.\\
An E move is executed to obtain \\
$(n-r, n , n-(r+1), n-(r+2), \ldots 3, 2, 1, n-(r-1), n-(r-2), \ldots n-1)$.\\
Finally, $(RE)^{r-1}$ is executed to obtain \\
$(n-r, n-(r-1), n-(r-2), \ldots n-1, n , n-(r+1), n-(r+2), \ldots 3, 2, 1)$ which is $K_{r+1,n}$. \\
Therefore, the total number of moves required to obtain $K_{r+1,n}$ from $K_{r,n}$ is $(r-1) + 1 + 2(r-1)=3r-2$.
\end{proof}
\begin{lemma}
\label{l2}
The number of moves required to obtain $K_{n,n}$ from $K_{n-2,n}$ is 3.
\end{lemma}
\begin{proof} According to the definition, $K_{n-2,n}$ is \\
$(3, 4, \ldots, n-1, n, 2, 1)$. Executing $R^2$ on $K_{n-2,n}$ yields \\
$(2, 1, 3, \ldots , n-1, n)$. Then executing an $E$ move yields
$(1, 2, 3\ldots , n-1, n)$ which is $K_{n,n}$. Therefore, three moves suffice to transform $K_{n-2,n}$ into $K_{n,n}$.
\end{proof}
\begin{theorem}
An upper bound for the number of moves required to sort $R_n$ with $LRE$ is $\frac{3}{2} n^2$.
\end{theorem}
\begin{proof}
Let $J(n)$ be the number of moves required to sort $R_{n}$ with $LRE$. 
According to Lemma \ref{l1}, the number of moves required to obtain $K_{r+1,n}$ from  $K_{r,n}$ is $3r-2$. Let $A(n)$ be the number of moves required to obtain $K_{n-2,n}$ from $K_{1,n}$ (which is $R_n$). Then
\begin{equation*}
\begin{split}
A(n) &= \sum\limits_{r=1}^{n-3} (3r-2) \\
&= 3\sum\limits_{r=1}^{n-3} r - (2(n-3))\\
&= \frac{3}{2} (n-2)(n-3) - 2n + 6 \\
&= \frac{3}{2} (n^2 - 5n +6) -2n + 6\\
&= \frac{3}{2} n^2 -\frac{15}{2} n + 9 -2n + 6\\
&= \frac{3}{2} n^2 - \frac{19}{2} n + 15\\
\end{split}
\end{equation*}
According to Lemma \ref{l2}, the number of moves required to obtain $K_{n,n}$ from $K_{n-2,n}$ is 3. Therefore,
\begin{equation*}
\begin{split}
J(n) &=A(n)+3\\
&= \frac{3}{2} n^2 - \frac{19}{2} n + 18
\end{split}
\end{equation*}
Therefore, the total number of moves required to sort $R_n$ with $LRE$ is $\frac{3}{2} n^2 - \frac{19}{2} n + 18$. Ignoring the lower order terms an upper bound for number of moves required to sort $R_n$ with $LRE$ is $\frac{3n^2}{2}$. This is the first non-trivial upper bound for the number of moves required to sort $R_n$ with $LRE$.
\end{proof}

\section{Algorithm LRE1}
\noindent\\
We designed Algorithm $LRE1$ in order to obtain the tighter upper bound for sorting $R_n$ with $LRE$. We define a permutation $K_{r,n}^{'} \in S_n$ as follows. The largest $r$ elements of $\Sigma$ i.e. $n-(r-1),n-(r-2),\ldots n$ are in sorted order. $K_{r,n}^{'}$ is obtained by concatenating sublists $(n-r,n-(r+1), \ldots 3,2,1)$ and $(n-(r-1),n-(r-2),\ldots n)$. $K_{r,n}$ and $K_{r,n}^{'}$ differ by the starting position of sublist $(n-(r-1),n-(r-2),\ldots n)$. The starting position of $(n-(r-1),n-(r-2),\ldots n)$ in $K_{r,n}$ is 1 whereas in $K_{r,n}^{'}$ it is $n-r+1$. Algorithm $LRE1$ first transforms $R_n$ into $K_{\floor{\frac{n}{2}},n}^{'}$. Then it transforms $K_{\floor{\frac{n}{2}},n}^{'}$ into $K_{n,n}^{'}$ which is $I_{n}$. Let $k=\floor{\frac{n}{2}}$ and $k'=n-k$. Let $J'(n)$ be the number of moves executed by Algorithm $LRE1$ to sort $R_{n}$.\\ \\
{Input: $R_n$. Output: $I_n$.\\
Initialization:$\forall{i}$ $\pi[i]=R_n[i]$. $k=\floor{\frac{n}{2}}$, $k'=n-k=\ceil{\frac{n}{2}}$\\ All moves are executed on $\pi$.
\begin{algorithm}
\begin{algorithmic}[1]
\caption{Algorithm LRE1}
\STATE \textbf{D1:}
\IF{$k \not= 1$}
    \STATE Execute E move
\ENDIF
\IF{$k \geq 3$}
    \STATE \textbf{D2:}
    \STATE Execute $(LE)^{k-2}$
    \STATE Execute $(RE)^{k-2}$
    \STATE Execute L move
    \STATE \textbf{D3:}
    \IF{$k \geq 4$}
        \FOR{$i\in(0,\ldots,k-1)$}
            \IF{$\pi[1]=(n-1)$}
                \STATE Execute E move
                \IF{$i\geq$1}
                    \STATE Execute $(RE)^{i-1}$
                \ENDIF
            \ENDIF
            \STATE Execute $(L)^{i}$
        \ENDFOR
    \ENDIF
\ENDIF
\STATE \textbf{D4:}
\IF{$k \not= 1$}
    \STATE Execute $(L)^{2}$
\ELSE
    \STATE Execute L move
\ENDIF
\STATE \textbf{D5:}
\STATE Execute E move
\IF{$k' \geq 3$}
    \STATE \textbf{D6:}
    \STATE Execute $(LE)^{k'-2}$
    \STATE Execute $(RE)^{k'-2}$
    \IF{$k' \geq 4$}
        \STATE \textbf{D7:}
        \STATE Execute L move
        \STATE \textbf{D8:}
        \FOR{$i\in(0,\ldots,k'-1)$}
            \IF{$\pi[1]=(k'-1)$}
                \STATE Execute E move
                 \IF{$i\geq$1}
                \STATE Execute $(RE)^{i-1}$
                \ENDIF
            \ENDIF
            \IF{$i \not= (k'-3)$}
                \STATE Execute $(L)^{i}$
            \ENDIF
        \ENDFOR
        \STATE \textbf{D9:}
        \STATE Execute R move
    \ENDIF
\ENDIF
\end{algorithmic}
\end{algorithm}
}
\subsection{Analysis}
\begin{lemma}
\label{l3}
The permutation obtained after executing D1 and D2 of Algorithm LRE1  is \\ $(n-1,n-2,\ldots,\ceil{\frac{n}{2}}+2, n, \ceil{\frac{n}{2}}, \ceil{\frac{n}{2}}-1, \ldots 3,2,1,\ceil{\frac{n}{2}}+1)$ and the number of moves executed is $2n-6$ when n is even and $2n-8$ when n is odd $\forall n\geq6$.
\end{lemma}
\begin{proof}
Execution of E move on $R_n$ in $D1$ yields \\
$(n-1,n,n-2,\ldots, 3,2,1)$.\\
Then executing $(LE)^{k-2}$ in  $D2$ yields \\ 
$(\ceil{\frac{n}{2}}+1, n, \ceil{\frac{n}{2}}, \ceil{\frac{n}{2}}-1, \ldots 3,2,1,n-1,n-2,\ldots,\ceil{\frac{n}{2}}+2)$.\\
Then executing $(RE)^{k-2}$ in $D2$ yields\\ $(\ceil{\frac{n}{2}}+1,n-1,n-2,\ldots,\ceil{\frac{n}{2}}+2, n, \ceil{\frac{n}{2}}, \ceil{\frac{n}{2}}-1, \ldots 3,2,1)$.\\
Then performing $L$ in  $D2$ move yields \\ $(n-1,n-2,\ldots,\ceil{\frac{n}{2}}+2, n, \ceil{\frac{n}{2}}, \ceil{\frac{n}{2}}-1, \ldots 3,2,1,\ceil{\frac{n}{2}}+1)$. \\
Therefore, the total number of moves executed in step $D1$ and $D2$ is \\ $1+4*(\floor{\frac{n}{2}}-2)+1=4\floor{\frac{n}{2}}-6=$
$\begin{cases}
2n-6 & \text{if $n$ is even}\\
2n-8 & \text{if $n$ is odd}
\end{cases}.\\
$
\end{proof}
\begin{lemma}
\label{l4}
The permutation obtained after D3 and D4 of LRE1 algorithm are executed is $K_{\floor{\frac{n}{2}},n}^{'}$ and the number of moves executed in the above two steps is $\frac{3n^2-34n+112}{8}$ when n is even and $\frac{3n^2-40n+149}{8}$ when n is odd $\forall n \geq 8$. 
\end{lemma}
\begin{proof}
From Lemma \ref{l3}, the permutation obtained after steps $D1$ and $D2$ is \\ $(n-1,n-2,\ldots,\ceil{\frac{n}{2}}+2, n, \ceil{\frac{n}{2}}, \ceil{\frac{n}{2}}-1, \ldots 3,2,1,\ceil{\frac{n}{2}}+1)$.\\
When $i=0$ in step $D3$ only $E$ move is executed and permutation thus obtained is \\
$(n-2,n-1,\ldots,\ceil{\frac{n}{2}}+2, n, \ceil{\frac{n}{2}}, \ceil{\frac{n}{2}}-1, \ldots 3,2,1,\ceil{\frac{n}{2}}+1)$.\\
When $i=1$ in step $D3$ only $L$ move is executed and permutation thus obtained is \\
$(n-1,\ldots,\ceil{\frac{n}{2}}+2, n, \ceil{\frac{n}{2}}, \ceil{\frac{n}{2}}-1, \ldots 3,2,1,\ceil{\frac{n}{2}}+1,n-2)$.\\
There after in each iteration in step $D3$, $E$ move, $(RE)^{i-1}$ and $L^{i}$ are executed so that the elements between $\pi[1]$ and $\pi[n-i+2]$ are left rotated. Thus, the permutation obtained after step $D3$ is $(n-1, n, \ceil{\frac{n}{2}}, \ceil{\frac{n}{2}}-1.\dots, 2,1,\ceil{\frac{n}{2}}+1,\ldots,n-2)$ and the number of moves executed in each iteration is $1+2(i-1)+i=3i-1$.\\
Therefore, the total number of moves executed in step $D3$ is \\ $1+1+\sum_{j=2}^{\floor{\frac{n}{2}}-3}(3j-1)\\=2+\sum_{i=2}^{\floor{\frac{n}{2}}-3}(3i-1)\\=$
$\begin{cases}
\frac{3n^2-34n+96}{8} & \text{if $n$ is even}\\
\frac{3n^2-40n+133}{8} & \text{if $n$ is odd}\\
\end{cases}$\\
Execution of $L^2$ in step $D4$ yields \\
$(\ceil{\frac{n}{2}}, \ceil{\frac{n}{2}}-1.\dots, 2,1,\ceil{\frac{n}{2}}+1,\ldots,n-2,n-1, n)$ which is $K_{\floor{\frac{n}{2}},n}^{'}$. Therefore, the total number of moves executed in steps $D3$ and $D4$ are $\begin{cases}
\frac{3n^2-34n+112}{8} & \text{if $n$ is even}\\
\frac{3n^2-40n+149}{8} & \text{if $n$ is odd}\\
\end{cases}$.\\
\end{proof}

\begin{lemma}
\label{l5}
The permutation obtained after executing D5 and D6 of LRE1 algorithm is \\ $(1,\ceil{\frac{n}{2}}-1,\ceil{\frac{n}{2}}-2,\ldots,2,\ceil{\frac{n}{2}},\ceil{\frac{n}{2}}+1,\ceil{\frac{n}{2}}+2,\ldots,n-1,n)$ and the number of moves executed in the above two steps is $2n-7$ when n is even and $2n-5$ when n is odd $\forall n\geq 5$. 
\end{lemma}
\begin{proof}
According to Lemma \ref{l4}, the permutation obtained after the steps $D1$ to $D4$ is \\
$(\ceil{\frac{n}{2}},\ceil{\frac{n}{2}}-1,\ceil{\frac{n}{2}}-2,\ldots,2,1,\ceil{\frac{n}{2}}+1,\ceil{\frac{n}{2}}+2,\ldots,n-1,n)$.\\
Now, executing $E$ move in step $D5$ yields \\ $(\ceil{\frac{n}{2}}-1,\ceil{\frac{n}{2}},\ceil{\frac{n}{2}}-2,\ldots,2,1,\ceil{\frac{n}{2}}+1,\ceil{\frac{n}{2}}+2,\ldots,n-1,n)$.\\
Then executing $(LE)^{k'-2}$ in step $D6$ yields\\ $(1,\ceil{\frac{n}{2}},\ceil{\frac{n}{2}}+1,\ceil{\frac{n}{2}}+2,\ldots,n-1,n,\ceil{\frac{n}{2}}-1,\ceil{\frac{n}{2}}-2,\ldots,2)$.\\
Then executing $(RE)^{k'-2}$ in step $D6$ yields\\ $(1,\ceil{\frac{n}{2}}-1,\ceil{\frac{n}{2}}-2,\ldots,2,\ceil{\frac{n}{2}},\ceil{\frac{n}{2}}+1,\ceil{\frac{n}{2}}+2,\ldots,n-1,n)$.\\
Therefore, the total number of moves executed in steps $D5$ and $D6$ is \\ $1+4(k'-2)=4k'-7=\begin{cases}
2n-7 & \text{if $n$ is even}\\
2n-5 & \text{if $n$ is odd}
\end{cases}$.\\
\end{proof}
\begin{lemma}
\label{l6}
The permutation obtained after executing D7 and D8 of LRE1 algorithm is $(2,3,\ldots,n-1,n,1)$ and the number of moves executed in the above two steps is $\frac{3n^2-38n+128}{8}$ when n is even and $\frac{3n^2-32n+93}{8}$ when n is odd $\forall n\geq 7$.
\end{lemma}
\begin{proof}
According to Lemma \ref{l5}, the permutation obtained after steps $D1$ to $D6$ is \\
$(1,\ceil{\frac{n}{2}}-1,\ceil{\frac{n}{2}}-2,\ldots,2,\ceil{\frac{n}{2}},\ceil{\frac{n}{2}}+1,\ceil{\frac{n}{2}}+2,\ldots,n-1,n)$.\\
$L$ move is executed in step $D7$ and the permutation thus obtained is \\ $(\ceil{\frac{n}{2}}-1,\ceil{\frac{n}{2}}-2,\ldots,2,\ceil{\frac{n}{2}},\ceil{\frac{n}{2}}+1,\ceil{\frac{n}{2}}+2,\ldots,n-1,n,1)$.\\
When $i=0$ in step D8, only $E$ move is executed and the permutation thus obtained is \\ $(\ceil{\frac{n}{2}}-2,\ceil{\frac{n}{2}}-1,\ldots,2,\ceil{\frac{n}{2}},\ceil{\frac{n}{2}}+1,\ceil{\frac{n}{2}}+2,\ldots,n-1,n,1)$. \\
When $i=1$ in step D8, only $L$ move is executed and the permutation thus obtained is \\ $(\ceil{\frac{n}{2}}-1,\ldots,2,\ceil{\frac{n}{2}},\ceil{\frac{n}{2}}+1,\ceil{\frac{n}{2}}+2,\ldots,n-1,n,1,\ceil{\frac{n}{2}}-2)$.\\
There after in each iteration in step $D8$ except when $i=(k'-3)$, $E$ move, $(RE)^{i-1}$ and $L^i$ are executed so that the elements between $\pi[1]$ and $\pi[n-i+2]$ are left rotated. When $i=k'-3$ only $E$ move and  $(RE)^{i-1}$ are executed. Thus, the obtained permutation after step $D8$ is $(2,3,\ldots,n-1,n,1)$. The number of moves executed in step $D8$ is \\ $1+1+1+\sum_{j=2}^{\ceil{\frac{n}{2}-4}}(3j-1)+1+2*\ceil{\frac{n}{2}-4}\\=4+\sum_{j=2}^{\ceil{\frac{n}{2}-4}}(3j-1)+2*\ceil{\frac{n}{2}-4}\\= \begin{cases}
\frac{3n^2-38n+128}{8} & \text{if $n$ is even}\\
\frac{3n^2-32n+93}{8} & \text{if $n$ is odd}
\end{cases}$.\\
\end{proof}
\begin{lemma}
\label{l7}
Algorithm LRE1 is correct.
\end{lemma}
\begin{proof}
According to Lemma \ref{l6}, the permutation obtained after steps $D1$ to $D8$ is \\ $(2,3,\ldots, \ceil{\frac{n}{2}}-1,\ceil{\frac{n}{2}},\ceil{\frac{n}{2}}+1,\ceil{\frac{n}{2}}+2,\ldots,n-1,n,1)$.\\
Executing $R$ move in step $D9$ yields\\
$(1,2,3,\ldots, \ceil{\frac{n}{2}}-1,\ceil{\frac{n}{2}},\ceil{\frac{n}{2}}+1,\ceil{\frac{n}{2}}+2,\ldots,n-1,n)$ which is $I_n$. Hence proves the lemma.
\end{proof}
\begin{theorem}
The number of moves required to sort $R_n$ with LRE1 algorithm is \\
J'(n) = $\begin{cases}
2 & \text{if $n=3$}\\
4 & \text{if $n=4$}\\
8& \text{if $n=5$}\\
13& \text{if $n=6$}\\
20& \text{if $n=7$}\\
\frac{3n^2-20n+72}{4}& \text{if $n\geq8$ and $n$ is even}\\
\frac{3n^2-20n+73}{4}& \text{if $n\geq8$ and $n$ is odd}
 \end{cases}$.
\end{theorem}
\begin{proof}
Case-i: $n$ is 3\\ \\
When $n$=3, the values of $k$ and $k'$ are 1 and 2 respectively. So, only steps $D1$ and $D4$ are executed. Therefore, the number of moves executed are $1+1=2$.\\\\
Case-ii: $n$ is 4\\ \\
When $n$=4, the values of both $k$ and $k'$ is 2, So only steps $D1$, $D4$ and $D5$ are executed. Therefore, the total number of moves executed are $1+2+1=4$.\\\\
Case-iii: $n$ is 5\\ \\
When $n$=5, the values of $k$ and $k'$ are 2 and 3 respectively. So only steps $D1$, $D4$, $D5$ and $D6$ are executed. According to Lemma \ref{l5}, the number of moves executed by steps $D5$ and $D6$ is $2n-5$ when $n$ is odd. Therefore, the total number of moves executed are $1+2+2n-5=2n-2=8$.\\\\
Case-iv: $n$ is 6\\ \\
When $n$=6, the values of both $k$ and $k'$ is 3. Therefore steps $D1$, $D2$, $D4$, $D5$ and $D6$ are executed. According to Lemma \ref{l3}, the number of moves executed by steps $D1$ and $D2$ is $2n-6$ when $n$ is even. According to Lemma \ref{l5}, the number of moves executed by steps $D5$ and $D6$ is $2n-7$ when $n$ is even. Therefore, the total number of moves executed are $2n-6+2+2n-7=4n-11=13$.\\\\
Case-v: $n$ is 7\\ \\
When $n$=7, the values of $k$ and $k'$ are 3 and 4 respectively. Therefore steps $D1$, $D2$, $D4$, $D5$ and $D6$ are executed. According to Lemma \ref{l3}, the number of moves executed by steps $D1$ and $D2$ is $2n-8$ when $n$ is odd. The number of moves executed by step $D4$ is 2. According to Lemma \ref{l5}, the number of moves executed by steps $D5$ and $D6$ is $2n-5$ when $n$ is odd. According to Lemma \ref{l6}, the number of moves executed by steps $D7$ and $D8$ is $\frac{3n^2-32n+93}{8}$ when $n$ is odd. Number of moves executed by step $D9$ is 1. Therefore, the total number of moves executed is $2n-8+2+2n-5+\frac{3n^2-32n+93}{8}+1=4n-11+\frac{3n^2-32n+93}{8}+1=20$.\\\\
Case-vi: $n\geq8$ and $n$ is even\\ \\
In this case all steps from $D1$ to $D9$ are executed. According to Lemma \ref{l3}, the number of moves executed by steps $D1$ and $D2$ is $2n-6$ when $n$ is even. According to Lemma \ref{l4}, the number of moves executed by steps $D3$ and $D4$ is $\frac{3n^2-34n+112}{8}$ when $n$ is even.  According to Lemma \ref{l5}, the number of moves executed by steps $D5$ and $D6$ is $2n-7$ when $n$ is even. According to Lemma \ref{l6}, the number of moves executed by steps $D7$ and $D8$ is $\frac{3n^2-38n+128}{8}$ when $n$ is even. Number of moves executed by step $D9$ is 1. Therefore, the total number of moves executed by Algorithm $LRE1$ is \\
\footnotesize 
\begin{equation*}
\begin{split}
J'(n) &=2n-6+\frac{3n^2-34n+112}{8}+2n-7+\frac{3n^2-38n+128}{8}+1\\
&= \frac{3n^2-20n+72}{4}
\end{split}
\end{equation*}\\
\normalsize
Case-vii: $n\geq8$ and $n$ is odd\\ \\
In this case all steps from $D1$ to $D9$ are executed. According to Lemma \ref{l3}, the number of moves executed by steps $D1$ and $D2$ is $2n-8$ when $n$ is odd. According to Lemma \ref{l4}, the number of moves executed by steps $D3$ and $D4$ is $\frac{3n^2-40n+149}{8}$ when $n$ is odd.  According to Lemma \ref{l5}, the number of moves executed by steps $D5$ and $D6$ is $2n-5$ when $n$ is odd. According to Lemma \ref{l6}, the number of moves executed by steps $D7$ and $D8$ is $\frac{3n^2-32n+93}{8}$ when $n$ is odd. Number of moves executed by step $D9$ is 1. Therefore, the total number of moves executed by Algorithm $LRE1$ is \\
\footnotesize
\begin{equation*}
\begin{split}
J'(n) &=2n-8+\frac{3n^2-40n+149}{8}+2n-5+\frac{3n^2-32n+93}{8}+1\\
&= \frac{3n^2-20n+73}{4}
\end{split}
\end{equation*}\\
\normalsize
Therefore, ignoring the lower order terms the new tighter upper bound for number of moves required to sort $R_n$ with $LRE$ is $\frac{3n^2}{4}$.
\end{proof}
\section{Exhaustive Search Results}
A branch and bound algorithm that employs BFS, i.e. \emph{Algorithm Search}, has been designed for computing the minimum number of moves to sort $R_n$ for a given $n$. It yielded values of 43 for $n=10$ and 53 for $n=11$. Thus, including the current values, the identified minimum number of moves for for $n=1 \ldots 11$ are respectively $(0, 1, 2, 4, 8, 13, 19, 26, 34, 43, 53)$. 
A list of permutations whose distance has been computed is maintained and the execution in every branch terminates either upon reaching $I_n$ or exceeding a bound. E, L and R generators are applied to each of the intermediate permutations yielding the corresponding permutations. 
We avoid application of two successive generators that are inverses of each other as such a sequence cannot be a part of optimum solution. Notation: \textit{Node} contains a permutation $\in S_n$ and its distance from $R_n$ corresponds to the minimum number of moves. With this algorithm and better computational resources one will be able to compute the corresponding values for larger values of $n$. \\\\
\noindent \textbf{Algorithm Search}\\
Initialization: The source vertex $\delta$ contains the permutation $R_n$ and its path is initialized to null. It is enqueued into BFS queue $Q$.\\
\noindent Input:$R_n$. Output: Optimum number of moves to reach $I_n$ with $LRE$.
\begin{algorithm}
\begin{algorithmic}[1]
\caption{Algorithm Search}
\WHILE{(Q is not empty)}
\STATE Dequeue $u$ from $Q$
\IF{($u$ is visited)}
\STATE continue
\ENDIF
\STATE Mark $u$ as visited
\IF{($u$ is $I_n$)}
\STATE \COMMENT{Array is sorted}
\RETURN length of $u.path$
\STATE break
\ENDIF

\IF{(Last move on $u.path$ $\not= E$ or $u.path =$ null  )}
\STATE Execute E on $u \rightarrow v$
\IF { $v$ is not visited}
\STATE  $v.path \gets u.path$ followed by E
\STATE Enqueue $v$ to Q
\ENDIF
\ENDIF
\IF{(Last move on $u.path$ $\not= L$ or $u.path =$ null  )}
\STATE Execute R on $u\rightarrow v$
\IF { $v$ is not visited}
\STATE $v.path \gets u.path$ followed by R
\STATE Enqueue $v$ to Q
\ENDIF
\ENDIF
\IF{(Last move on $u.path$ $\not= R$ or $u.path =$ null  )}
\STATE Execute L on $u\rightarrow v$
\IF { $v$ is not visited}
\STATE $v.path \gets u.path$ followed by L
\STATE Enqueue $v$ to Q
\ENDIF
\ENDIF
\ENDWHILE
\end{algorithmic}
\end{algorithm}
\section{Results}
Comparison of the number of moves required to sort $R_n$ with $LRE$ by various algorithms. The first column shows $n$, the size of permutation. Subsequent columns show the number of moves required to sort $R_n$ with Algorithms $LRE$, $LRE1$ and $Search$ respectively.
\begin{table}[htbp]
\begin{center}
\begin{tabular}{|c|c|c|c|}
\hline
\textbf{\textit{n}} & \textbf{\textit{LRE}}& \textbf{\textit{LRE1}}& \textbf{\textit{Search(Optimal)}} \\
\hline
3&3&2&2\\
4&4&4&4\\
5&8&8&8\\
6&15&13&13\\
7&25&20&19\\
8&38&26&26\\
9&54&34&34\\
10&73&43&43\\
11&95&54&53\\
\hline
\end{tabular}
\label{tab1}
\end{center}
\end{table}
\section*{Conclusion}
The first known upper bound for sorting $R_n$ with $LRE$ has been shown. A tighter upper bound has been derived. The future work consists of identifying the exact upper bound for sorting $R_n$ with $LRE$. 
The identification of the diameter of the $LRE$ Cayley graph and the characterization of permutations that are farthest from $I_n$ in this Cayley graph are open questions.
\section*{Acknowledgement} Venkata Vyshnavi Ravella and CK Phani Datta have been contributing towards implementation, testing and analysis of the algorithms. Authors thank Jayakumar P for helpful suggestions. Authors thank Tony Bartoletti for personal communication.

\end{document}